\documentclass{birkjour}
\usepackage{amsmath,amsfonts, amsthm,amssymb,hyperref,bbm,eucal,verbatim}

\newtheorem{thm}{Theorem}
\newtheorem{lemma}{Lemma}[section]
\newtheorem{prop}[lemma]{Proposition}

\newtheorem{cl}[lemma]{Claim}
\newtheorem{df}[lemma]{Definition}

\begin{document}

\title{Resonant delocalization on the Bethe strip}
\author{Mira Shamis}
\address{Department of Mathematics, Princeton University, Princeton,
NJ 08544, USA}
\email{mshamis@princeton.edu}
\thanks{This research was supported by NSF grant PHY-1104596.}

\begin{abstract}
Recently, Aizenman and Warzel discovered a mechanism for the appearance of absolutely continuous spectrum 
for random Schr\"odin\-ger operators on the Bethe lattice through rare resonances (resonant delocalization).
We extend their analysis to operators with matrix-valued random potentials drawn from ensembles such as the Gaussian
Orthogonal Ensemble. These operators can be viewed as random operators on the Bethe strip, 
a graph (lattice) with loops. 
\end{abstract}
\maketitle

\section{Introduction}

Let $\mathcal{T}$ be a regular rooted tree with branching number $K > 1$ (Bethe lattice). We shall be interested 
in random Schr\"odinger operators on the Cartesian product $\mathcal{T} \times G$ of $\mathcal{T}$
and a finite graph $G$ with $W$ vertices (Bethe strip). Equivalently, these can be seen as random Schr\"odinger
operators on $\mathcal{T}$ with matrix-valued potential. The precise definition is as follows:
$H = H_{\lambda,\omega}$ is a random operator acting on 
\[ \ell^2(\mathcal{T} \times G) = \ell^2(\mathcal{T} \to \mathbb{R}^W)~, \] 
and given by the matrix elements
\begin{equation}\label{eq:def.H}
H_{\lambda,\omega} (x, y) = \begin{cases}
				  \mathbbm{1}_{W \times W}~, & x \sim y  \,\,(\text{$x$ is adjacent to $y$})\\
				  A + \lambda V_\omega(x)~.  & x = y \\
				  0~, &\text{otherwise}
                               \end{cases}~,
\quad x,y \in \mathcal{T}~. 
\end{equation}
Here $\lambda \geq 0$ is a coupling constant, $\omega$ denotes an element of the probability space, 
$A$ is a fixed $W \times W$ Hermitian matrix, and $V_\omega(x)$ are independent identically distributed
$W \times W$ random matrices. The potential $A + \lambda V_\omega(x)$ will be denoted $U_\omega(x)$.

The question that we shall address is, what is the spectral type of $H$ when $\lambda$ is small. Before 
stating our results, let us review what was previously known. 

\vspace{1mm}\noindent
For the Bethe lattice ( $W=1$, $A = 0$ in our notation), the spectrum of the unperturbed operator ($\lambda = 0$) is purely absolutely
continuous and fills the interval $[-2 \sqrt{K}, 2\sqrt{K}]$. Under mild assumptions on the potential, Klein
showed \cite{Kl1,Kl2,Kl3} that, for small $\lambda > 0$, the spectrum in $[-2\sqrt{K}+\epsilon, 2\sqrt{K}-\epsilon]$
is also (almost surely) absolutely continuous. Additional proofs and generalizations of this result were found by Aizenman,
Sims, and Warzel \cite{ASW}, and by Froese, Hasler, and Spitzer \cite{FHS}.

On the other hand, Aizenman proved \cite{A} that, for small $\lambda$, the spectrum of $H$ outside
$[-K-1-\epsilon, K+1+\epsilon]$ is almost surely pure point. 

In the recent work \cite{AW}, Aizenman and Warzel proved the presence of absolutely continuos spectrum thoroughout
the interval $[-K-1+\epsilon, K+1-\epsilon]$. They found a new mechanism for the appearance of absolutely
continuous spectrum, entirely different from the one appearing inside the spectrum of the unperturbed operator, and
coined the term ``resonant delocalization'' for it. As opposed to the absolutely continuous spectrum in the
interval $[-2\sqrt{K}, 2\sqrt{K}]$, which appears due to the stability of the absolutely continuous spectrum
on the Bethe lattice, the absolutely continuous spectrum in $[-K-1,K+1] \setminus [-2\sqrt{K}, 2\sqrt{K}]$
(in the Lifshitz tails) appears due to resonances between distant sites. The interval $[K-1,K+1]$ is exactly
the $\ell^1$ spectrum of the unperturbed operator; the importance of the $\ell^1$ spectrum is further discussed
in \cite{AW} and in the survey \cite{W} by Warzel.

\vspace{2mm}\noindent
The goal of this present work is to extend the result of \cite{AW} to the case $W > 1$ of the Bethe strip.
We make use of significant parts of the work \cite{AW}; for the reader's convenience, we denote by
Statement~X* the generalization of \cite[Statement~X]{AW}.

Denote by $\{\nu_i\}_{i=1}^W$ the eigenvalues of $A$, and let 
\[ S_\epsilon  = \bigcup_i \left[ \nu_i - (K+1) + \epsilon, \, \nu_i + (K+1) - \epsilon \right]~. \]
Our main result is 
\begin{thm}[Corollary~2.3*]\label{thm:main}
Assume that $V_\omega(x)$ are drawn from the Gaussian Orthogonal Ensemble (GOE). For any $\epsilon > 0$
any open interval $I \subset S_\epsilon$ almost surely has absolutely continuous spectrum of $H_{\lambda,\omega}$
in it, when $\lambda>0$ is sufficiently small.
\end{thm}

Thus the mechanism of resonant delocalization from \cite{AW} may be extended to
the Bethe strip, a lattice with loops. See \cite[Section 4]{W} for a more general discussion of possible
further extensions.

\vspace{2mm}\noindent

Theorem \ref{thm:main} should also be compared with the result of Klein and Sadel \cite{KS} (and its ramification
\cite{KS2}), who proved, under weaker assumptions on the potential $V_\omega$, that the spectrum of 
$H_{\lambda,\omega}$ in
\[ S_\epsilon^-  = \bigcap_i \left[ \nu_i - 2\sqrt{K} + \epsilon, \, \nu_i + 2\sqrt{K} - \epsilon \right] \]
is almost surely purely absolutely continuous; the special case $K = W = 2$ was earlier considered by Froese,
Halasan, and Hasler \cite{FHH}. Thus we replace the intersection with union (i.e.\ the fastest
Lyapunov exponent with the slowest one) and $2\sqrt{K}$ with $K+1$ (i.e.\ the $\ell^2$ spectrum with the $\ell^1$
spectrum) at the price of more restrictive assumptions on $V_\omega$, and we only manage to show the existence
of absolutely continuous spectrum rather than its purity. The spectrum outside the set $S^-_{\epsilon}$ is pure point, 
as follows from the results of \cite{A}. Thus our result provides an additional example of the appearance of
absolutely continuous spectrum in the $\ell^1$ spectrum of the unperturbed operator $H_{0,\omega}$, well outside 
the $\ell^2$ spectrum.

\vspace{2mm}\noindent
Theorem~\ref{thm:main} will follow from Theorems~\ref{thm1} and \ref{thm2} below. Theorem~\ref{thm2} connects the presence of 
absolutely continuous spectrum with the (slowest) Lyapunov exponent $L = L_\lambda(E) \in \mathbb{R}_+$, 
which is defined in the sequel. Theorem~\ref{thm1}, which holds for any (independent identically distributed) 
random potential $U_\omega$  with $\mathbb{E} \log^+ \|U_\omega(x)\| < \infty$, guarantees that the assumptions of 
Theorem~\ref{thm1} are satisfied for small $\lambda$.

\begin{thm}\label{thm1}
For every $\epsilon > 0$ and any interval $I \subset S_\epsilon$ one has
\[ \mathrm{mes} \left\{ E \in I \, \mid \, L(E) < \log K \right\} > 0 \]
for sufficiently small $\lambda$.
\end{thm}

It is probably true that for $\lambda < \lambda_0(\epsilon)$ one has $L|_{S\epsilon} < \log K$; this is
however unsettled even for $W=1$ (except for the special case of Cauchy disorder, see \cite{AW}).

\vspace{2mm}\noindent
In the next two theorems, we assume that $V_\omega(x)$ are drawn from the Gaussian Orthogonal Ensemble (GOE).
We shall comment on possible generalizations in the sequel.
\begin{thm}[Theorem~2.1*]\label{thm2}
The absolutely continuous spectrum of $H$ fills (almost surely) the set $\left\{ E \, \mid \, L(E) < \log K \right\}$,
meaning that the restriction of the Lebesgue measure to this set is almost surely absolutely continuous with respect
to the absolutely continuous part of the spectral measure of $H$. In particular, this set is a subset of the absolutely
continuous spectrum of $H$.
\end{thm}

Similarly to the results of \cite{AW}, Theorem~\ref{thm2} is sharp in the following sense: the spectrum of $H_{\lambda,\omega}$
in $\left\{ E \, \mid \, L_\lambda(E) > \log K \right\}$ is almost surely pure point, as follows from the results of \cite{A}.

\vspace{2mm}\noindent
For expositional reasons, we first prove
\begin{thm}\label{thm2'}
$H$ has (almost surely) no pure point spectrum in the set
\[ \left\{ E \, \mid \, L(E) < \log K \right\}~. \]
\end{thm}
\noindent and then the stronger Theorem~\ref{thm2}.

\vspace{2mm}\noindent
Finally, let us comment on the generality of the results. The simplest generalization of the Bethe 
strip setting of \cite{AW} is the GOE potential, corresponding to $A = 0$ (and small $\lambda > 0$). 
In this case, only minor modifications (due to the non-commutativity of matrix product) would be required
in the arguments of \cite{AW}, since the Lyapunov exponents differ from one another by a quantity which
vanishes in the limit $\lambda \to 0$ (at least, in the sense of Theorem~\ref{thm1}).

When $A \neq 0$, additional difficulties arise, which are due to the fact that there may be a significant
difference between the fastest and the slowest Lyapunov exponent. Most of the current paper is devoted to
overcoming these difficulties. We state the results for the case when $V_\omega(x)$ are drawn from the Gaussian
Orthogonal Ensemble, but the arguments may be extended to more general potentials with off-diagonal disorder.
The crucial requirement is the {\em conditional a.c.\ property}, stating that the conditional distribution
of $V_\omega(x)_{i_0, j_0}$ given $\{ V_\omega(x)_{ij} \, \mid \, (i, j) \neq (i_0, j_0), (j_0, i_0)\}$ is absolutely continuous.
We try to indicate where the off-diagonal disorder assumption is used in the proof.

It would be interesting to extend the results of this paper to the case of diagonal disorder: for example, $V_\omega(x)$ is
a diagonal matrix with independent identically distributed entries (which would correspond to the usual Bethe strip).

\section{Preliminaries and proof of Theorem~\ref{thm1}}

For 
\[ z \in \mathbb{C}^+ = \left\{ z \in \mathbb{C} \, \mid \, \Im z > 0\right\}~, \]
the Green function $G_\lambda(x, y; z)$ is the $xy$ block of the
resolvent $(H_\lambda - z)^{-1}$ (from this point we suppress the dependence on $\omega$). 
For a vertex $u$ of $\mathcal{T}$, $G_\lambda^{\mathcal{T}_u}(x, y, z)$ is the $xy$ block
of the Green function associated with the restriction of $H_\lambda$ to the subgraph $\mathcal{T}_u$ obtained by
removing $u$ from $\mathcal{T}$. $\mathcal{N}_u^+$ is the collection of forward neighbors of a vertex $u$, and 
$\mathcal{N}_u$ is the collection of all neighbors of $u$. The root of $\mathcal{T}$ is denoted $0$.

\begin{cl}[Proposition~3.1*]\label{cl:1ident} For any matrix-valued Schr\"odinger operator $H$ on $\mathcal{T}$ 
with potential $U$, and any $z \in \mathbb{C}^+$,
\[ G_\lambda(x,x; z) = \left(U(x) - z - \sum_{y \in \mathcal{N}_x} G^{\mathcal{T}_x}(y,y; z)\right)^{-1}~,\]
and for any ordered pair $0 \prec x \prec y$ 
\[\begin{split}
G_\lambda(x,y; z) 
&= G_\lambda(x,x; z) G_\lambda^{\mathcal{T}_x}(x_1, y; z) 
= G_\lambda^{\mathcal{T}_y}(x, x_n; z) G_\lambda(y,y; z)\\
&= G_\lambda(x,x; z) G_\lambda^{\mathcal{T}_{x}}(x_1, x_1; z) \cdots G_\lambda^{\mathcal{T}_{x_n}}(y,y; z)~,
\end{split}\]
where $x x_1 x_2 \cdots x_n y$ is the path from $x$ to $y$.
\end{cl}

\begin{proof}
To prove the first statement, decompose
\[ \ell_2(\mathcal{T} \to \mathbb{R}^W) = \ell^2 (\{x\} \to \mathbb{R}^W) \oplus \ell_2(\mathcal{T}_x \to \mathbb{R}^W)~, \]
and apply the Schur--Banachiewicz formula for block matrix inversion. To prove the second statement,
we iterate the formula
\[ G_\lambda(x, y; z) = G_\lambda(x, x; z) G_\lambda^{\mathcal{T}_x}(x_1, y; z)\]
which follows from the resolvent identity.
\end{proof}

Let $0 x_1 x_2 x_3 \cdots x_n \cdots$ be a branch of $\mathcal{T}$. Denote
\[ L(z) = - \lim_{n \to \infty} \frac{1}{n+1} \ln \| G_\lambda(0, x_n; z) \|~, \]
where $\| \cdot \|$ stands for the operator norm. This is the slowest Lyapunov exponent.

\begin{cl}\label{cl:2fk}
The Lyapunov exponent $L(z)$ is defined and non-random for any independent identically distributed  
matrix potential $U(x)$ which satisfies 
\[ \mathbb{E} \log^+ \| U(x) \| < \infty~.\]
\end{cl}

The claim follows from the Furstenberg--Kesten theorem \cite{FK}. For $U = A  +\lambda V$, we denote
the Lyapunov exponent by $L_\lambda$ when we need to emphasize the dependence on $\lambda$. For 
$E \in \mathbb{R}$, we set 
\[ L_\lambda(E) = \lim_{\eta \to +0} L_\lambda(E+i\eta)~. \]

\begin{cl}\label{cl:3conv}
For any matrix potential $U = A + \lambda V$, where $A$ is fixed and $V(x)$ are independent and identically
distributed with $\mathbb{E} \log^+ \|V(x)\|<\infty$, and for any $z \in \mathbb{C}^+$,
\[ L_\lambda(z) \to L_0(z) \quad \text{as} \quad {\lambda \to 0}~. \]
\end{cl}

Claim~\ref{cl:3conv} follows from the strong resolvent convergence outside the spectrum. From Claim~\ref{cl:3conv} 
and the Fatou lemma, we obtain
\begin{cl}\label{cl:6.1}[Theorem~6.1*]
For any matrix potential $U = A + \lambda V$, where $A$ is fixed and $V(x)$ are independent and identically
distributed, and for any bounded interval $I \subset \mathbb{R}$, the function
\[ \lambda \mapsto \int_I L_\lambda(E) dE \] 
is continuous, and, in particular,
\[ \lim_{\lambda \to 0} \int_I L_\lambda(E) dE = \int_I L_0(E) dE~. \]
\end{cl}

The argument justifying Claims~\ref{cl:3conv} amd \ref{cl:6.1} is identical to that of \cite[Section~6.1]{AW}.
Theorem~\ref{thm1} is a consequence of Claim~\ref{cl:6.1} and the explicit computation of the free Lyapunov
exponent $L_0$, which can be performed using Claim~\ref{cl:1ident} and  which shows that
\[ L_0(E) < \log K \iff E \in S_0 \equiv \bigcup_i (\nu_i - (K+1), \nu_i + (K+1))~. \]

\section{Proof of Theorem~\ref{thm2'}}

The proof of Theorem~\ref{thm2'} makes use of the following version of the Simon--Wolff criterion \cite{SW}:
\begin{prop}\label{prop1:sw}[Matrix Simon--Wolff criterion]
Suppose an i.i.d.\ matrix potential $U(x)$ satisfies the following two properties:
\begin{enumerate}
 \item $U(x)$ has independent entries on the diagonal,
 \item $U(x)$ is irreducible, meaning that it has no non-trivial deterministic invariant
subspace.
\end{enumerate}
Then the pure point part of the spectral measure is almost surely supported on the set
\[ \Sigma = \left\{ E \in \sigma(H) \, \mid \, \sum_{x \in \mathcal{T}} \|G(0, x; E+i0) \|^2 < \infty \quad \text{almost surely}\right\}~,\]
and the continuous part is almost surely supported on its complement.
\end{prop}

\begin{proof}
By the usual Simon--Wolff criterion \cite{SW}, the continuous spectrum is almost surely supported on
the set 
\[ S_j = \left\{ \sum_x \sum_i |G(0, x; E+i0)_{j,i}|^2 = \infty \right\}~,\]
and the pure point spectrum is almost surely supported on its complement. By assumption~2.,
the set $S_j$ is (almost surely) independent of $j$. Therefore it coincides with
\[ \left\{ \sum_x \sum_{ij} |G(0, x; E+i0)_{j,i}|^2 = \infty \right\}~,\]
and the latter coincides with
\[ \left\{ \sum_x \| G(0, x; E+i0)\|^2 = \infty \right\} \]
due to equivalence between norms.
\end{proof}

Now, Claim~\ref{cl:1ident} yields
\[ \|G(0, x; z) \| = \| G(x, x; z)^* G^{\mathcal{T}_x}(0, x_-; z)^* \| 
  \geq \|G(x, x;z)^* G^{\mathcal{T}_x}(0, x_-;z)^* w \|\]
for any unit vector $w$ (from this point we suppress the dependence on $\lambda$, and $x_-$ stands for the
backward neighbor of a vertex $x$). 
Let $v =  G^{\mathcal{T}_x}(0, x_-;z)^* w$ and $\tilde{v} = v / \|v\|$. Then
\begin{equation}\label{eq:bdin2'}
\begin{split}
\|G(0, x;z) \| 
  &\geq \|G(x, x;z)^* v \| \\
  &\geq | \langle G(x, x;z)^* v, \tilde{v} \rangle|
  = \| v \| \, | \langle G(x, x;z) \tilde{v}, \tilde{v} \rangle|~. 
\end{split}
\end{equation}
Let 
\[ w = w_{\max}(G^{\mathcal{T}_x}(0, x_-;z) G^{\mathcal{T}_x}(0, x_-;z)^*) \]
be the unit eigenvector of $G^{\mathcal{T}_x}(0, x_-;z) G^{\mathcal{T}_x}(0, x_-;z)^*$ associated with the largest eigenvalue; then 
$\tilde{v} = w_{\max}(G^{\mathcal{T}_x}(0, x_-;z)^*G^{\mathcal{T}_x}(0, x_-;z))$.
Denote
\[ E_x = \left\{ | \langle G(x, x; E + i \eta) \tilde{v}, \tilde{v} \rangle| \geq \tau \equiv e^{+ (L(E) + 2\delta) n} \right\}~,\]
\[ R_x = \left\{ \|G^{\mathcal{T}_{x}}(0, x_-; E+i\eta) \| \geq e^{- (L(E) + \delta) n} \right\}~, \] 
and 
\[ N = \sum_{x \in S_n} \mathbbm{1}_{R_x \cap E_x}~,\]
where $S_n = \mathcal{N}_+^n (0)$ is the sphere of radius $n$ about the root.                          
According to (\ref{eq:bdin2'}), 
\[ \|G(0,x; E+i\eta)\| \geq e^{\delta n} \quad \text{on} \quad R_x \cap E_x~. \]

\begin{prop}[First moment bound]\label{prop:2firstmoment} For $U(x) = A + \lambda V(x)$, where $V(x)$ are drawn from
the Gaussian Orthogonal Ensemble, 
\[ \mathbb{E} N \geq \frac{1}{C(\lambda) \tau} K^n \]
when $n$ is large enough and $\eta > 0$ is small enough.
\end{prop}

\begin{proof}
By continuity in $\eta \to +0$ which holds for almost every energy (cf.\ \cite[Corollary~4.10]{AW}),
it is sufficient to prove the statement for $E + i0$.

Denote by $P$ the projection on 
\[ \tilde{v} =  w_{\max}(G^{\mathcal{T}_x}(0, x_-; E+i\eta)^*G^{\mathcal{T}_x}(0, x_-;E+i\eta))~; \]
$\tilde{v}$ is independent of $V(x)$. Also set $Q = \mathbbm{1} - P$. By Claim~\ref{cl:1ident},
\begin{multline}
\langle G(x, x; E + i \eta) \tilde{v}, \tilde{v} \rangle \\
  = P \left(A + \lambda V(x) - E - i\eta - \sum_{y \in N_x} G^{\mathcal{T}_x}(y,y; E + i\eta)\right)^{-1} P~.
\end{multline}
By the Schur--Banachiewicz formula
\[ PT^{-1}P = (P T P - P T Q (Q T Q)^{-1} Q T P )^{-1}~, \]
we have
\[ \langle G(x, x; E + i \eta) \tilde{v}, \tilde{v} \rangle = (g - \sigma)^{-1}~,\]
where $g = \lambda P V(x) P$ is Gaussian, and 
\begin{multline}\label{eq:sigma}
\sigma = -PAP + z + \sum_{y \in N_x} P G^{\mathcal{T}_x}(y,y; E+i\eta) P  \\
  + \left( P U(x) Q - \sum_{y \in N_x} P G^{\mathcal{T}_x}(y,y; E+i\eta) Q \right) \\
    \left( Q U(x) Q - z - \sum_{y \in N_x} Q G^{\mathcal{T}_x}(y,y; E+i\eta) Q\right)^{-1} \\
    \left( Q U(x) P - \sum_{y \in N_x} Q G^{\mathcal{T}_x}(y,y; E+i\eta) P \right)~.
\end{multline}

\begin{lemma}\label{l:1}
The random variable $\sigma$ is independent of $g$.
\end{lemma}

\begin{proof}(Uses off-diagonal randomness)
This fact is an immediate corollary of the following property of the Gaussian Orthogonal
Ensemble: for every orthogonal projection $P$, $PV(x)P$ is independent of 
\[ \left\{ (1-P)V(x)P, PV(x)(1-P), (1-P)V(x)(1-P)\right\}~. \]
\end{proof}

\begin{lemma}\label{l:2}
There exists $0 < s < 1$ so that
\[ \mathbb{E} |\sigma|^s \leq C~, \]
where $C > 0$ is a constant. 
\end{lemma}

\begin{proof}
We bound the $s$-moment of every term in (\ref{eq:sigma}). The bound on
\[ \mathbb{E} \left| \sum_{y \in N_x} P G^{\mathcal{T}_x}(y,y; E+i\eta) P \right|^s \]
follows from \cite[A.1]{AW}. It therefore remains to bound the $s$-moment of the multipliers in
(\ref{eq:sigma}) (then the $s/3$-moment of the product is bounded by Cauchy--Schwarz). The expressions
\[ \mathbb{E} \|PV(x)Q\|^s~, \quad \mathbb{E} \|QV(x)P\|^s\]
are estimated directly (they are finite e.g.\ for $s = 2$); the $s$-moment of the second multiplier in
(\ref{eq:sigma}) can be bounded using an argument similar to the upper bound in Lemma~\ref{l:3} below.
\end{proof}

Having the two lemmata, we can conclude the proof of Proposition~\ref{prop:2firstmoment}.
By Chebyshev's inequality and Lemma~\ref{l:2}, 
\begin{equation}\label{eq:cheb} \mathbb{P} \left\{ |\sigma|\leq t \right\} \geq 1 - C' / t^s \end{equation}
can be made arbitrarily close to $1$ by choosing $t$ large enough. Now we estimate $\mathbb{E} N$ as follows:
first,
\[ \mathbb{E} N = \sum_{x \in S_n} \mathbb{P}(R_x \cap E_x) = K^n \mathbb{P}(R_x \cap E_x)~. \]
Then
\[\begin{split}
 \mathbb{P}(R_x \cap E_x)
  &= \mathbb{P} \left( R_x \cap \big\{ |\lambda g - \sigma| \leq \tau^{-1} \big\} \right) \\
  &\geq \mathbb{P} \left( R_x \cap \{ |\sigma| \leq t \} \cap \big\{ |\lambda g - \sigma| \leq \tau^{-1} \big\} \right) \\
  &= \mathbb{E} \left( \mathbbm{1}_{R_x} \mathbbm{1}_{|\sigma|\leq t} \,\, 
    \mathbb{P} \left\{ |g - \sigma|\leq \frac{1}{\lambda \tau} \, \mid \, R_x , \sigma \right\}  \right)~.
  \end{split}
\]
From Lemma~\ref{l:1},
\[ \mathbb{P} \left\{ |g - \sigma|\leq \frac{1}{\lambda \tau} \, \big| \, R_x , \sigma \right\}  
  \geq \frac{1}{C_{\lambda,t} \tau} \mathbbm{1}_{|\sigma|\leq t}~, \]
therefore
\[ \mathbb{P}(R_x \cap E_x) \geq  \frac{1}{C_{\lambda,t} \tau} \mathbb{P} \left( R_x \cap \{ |\sigma|\leq t \} \right)~.\]
Choosing $n$ and $t$ large enough, we get 
\[ \mathbb{P} (R_x) \geq 3/4 \]
from Claim~\ref{cl:2fk} and 
\[ \mathbb{P}  \{ |\sigma|\leq t \} \geq 3/4~, \]
from (\ref{eq:cheb}), hence  
\[ \mathbb{P} \left( R_x \cap \{ |\sigma|\leq t \} \right) \geq 1/2 \]
and 
\[\mathbb{P}(R_x \cap E_x) \geq \frac{1}{2 C_{t,\lambda} \tau}~.\]
\end{proof}

Next, we bound the second moment of $N$ from above. The first ingredient is
\begin{lemma}\label{l:3}
For $s \in (0, 1)$,
\[ {C_-^{-1}(s, z)} \leq \frac{\mathbb{E} \|G^{\mathcal{T}_x}(0, x_-; z)\|^s}
                       {\mathbb{E} \| G^{\mathcal{T}_{u,x}}(0, u_-; z)\|^s \mathbb{E} \|G^{\mathcal{T}_{u,x}}(u_+, x_-; z) \|^s}
  \leq C_+(s,z)~,\]
where $C_\pm(s,z)$ are uniformly bounded as $\Im z \to +0$.
\end{lemma}

\begin{proof}
We start from Claim~\ref{cl:1ident}:
\begin{equation}\label{eq:ident}
G^{\mathcal{T}_x}(0, x_-; z) = G^{\mathcal{T}_{u,x}}(0, u_-; z) G^{\mathcal{T}_x}(u,u; z) G^{\mathcal{T}_{u,x}}(u_+, x_-; z)~. 
\end{equation}

\vspace{2mm}
{\noindent \bf Upper bound} (Only requires diagonal randomness) Taking norms in (\ref{eq:ident}), we obtain
\[ \| G^{\mathcal{T}_x}(0, x_-; z) \|^s \leq
\| G^{\mathcal{T}_{u,x}}(0, u_-; z) \|^s \| G^{\mathcal{T}_x}(u,u; z) \|^s \|  G^{\mathcal{T}_{u,x}}(u_+, x_-; z) \|^s~.  \]
By construction, $G^{\mathcal{T}_{u,x}}(0, u_-; z)$, $G^{\mathcal{T}_{u,x}}(u_+, x_-; z)$, and $V(u)$ are independent. We shall show 
that
\begin{equation}\label{eq:weg}
\mathbb{E}_{V(u)} \|G^{\mathcal{T}_x}(u,u; z) \|^s \leq C_+(s, z),
\end{equation}
where $\mathbb{E}_{V(u)}$ denotes averageing over $V(u)$ (= conditioning on all the other values of the potential).
Averaging (\ref{eq:weg}) over $\{V(y) \, \mid \, y \neq u \}$, we obtain the upper bound in the lemma. To prove (\ref{eq:weg}),
note that, by the Schur--Banachiewicz formula, 
\[ G^{\mathcal{T}_x}(u,u; z) = (\lambda V(u) - \sigma)^{-1}~, \]
where $\sigma$ is independent of $V(u)$. Therefore
\[ \mathbb{E}_{V(u)} \| G^{\mathcal{T}_x}(u,u; z) \|^s 
  \leq C \lambda^{-s} \sum_{j,k} \mathbb{E} |(V(u) - \sigma)^{-1}_{jk}|^s  = C (I + II)~,\]
where $I$ is the sum of the diagonal terms, and $II$ is the sum of the off-diagonal terms. To bound the diagonal terms,
note that
\[ \mathbb{E}_{V(u)} |(V(u) - \sigma)^{-1}_{jj}|^s = \mathbb{E}_{V(u)} \mathbb{E}_{V(u)_{jj}} |V(u)_{jj} - \tilde{\sigma}|^{-s}~, \]
where $\tilde{\sigma}$ is independent of $V(u)_{jj}$.  Therefore (by the inequality (II.2) from the paper of Aizenman--Molchanov \cite{AM})
\[ \mathbb{E}_{V(u)} |(V(u) - \sigma)^{-1}_{jj}|^s \leq C(s)  \]
and $I \leq C(s) W$.

To bound the off-diagonal terms,  we use inequality (II.3) from \cite{AM}. This concludes the proof of the upper bound.

\vspace{2mm}\noindent
{\bf Lower bound} (Uses off-diagonal randomness) We shall use

\begin{prop}\label{prop:goe}
Let $V$ be a random matrix drawn from GOE, and let $\sigma$ be a fixed matrix. Then for any two vectors
$\phi$ and $\psi$
\[ \mathbb{E} \left| \langle (V - \sigma)^{-1} \phi, \psi \rangle \right|^s \geq C_{\|\sigma\|,s} \|\phi\|^s \, \|\psi\|^s~.\]
\end{prop}

\begin{proof}
We may assume without loss of generality that $\phi = e_1$ (the first vector of the standard basis) and
that $\psi = a e_1 + b e_2$, $a^2 + b^2 = 1$. Then
\[ \langle (V - \sigma)^{-1} \phi, \psi \rangle  
  = a (V-\sigma)^{-1}_{11} + b (V-\sigma)^{-1}_{12}~.\]
By Cramer's rule,
\[ a (V-\sigma)^{-1}_{11} + b (V-\sigma)^{-1}_{12}
  = \frac{a (g_{22} - \tilde{\sigma}_{22}) - b (g_{12} - \tilde{\sigma}_{12})}
         {(g_{11} - \tilde{\sigma}_{11})(g_{22} - \widetilde{\sigma}_{22}) - (g_{12} - \tilde\sigma_{12})(g_{12} - \tilde\sigma_{21})}~, \]
where $g_{ij}$ are Gaussian, and $\tilde\sigma$ is independent of the $g_{ij}$. By H\"older's inequality,
\begin{multline*}
\mathbb{E}_g \left| a (V-\sigma)^{-1}_{11} + b (V-\sigma)^{-1}_{12} \right|^s \\
  \geq \frac{\left[ \mathbb{E}_g \left| a (g_{22} - \tilde{\sigma}_{22}) - b (g_{12} - \tilde{\sigma}_{12}) \right|^{s/2} \right]^2}
            {\mathbb{E}_g \left| (g_{11} - \tilde{\sigma}_{11})(g_{22} - \widetilde{\sigma}_{22}) - 
		(g_{12} - \tilde\sigma_{12})(g_{12} - \tilde\sigma_{21}) \right|^s} 
\end{multline*}
It is easy to see that the denominator is bounded from above by a number depending only on $\tilde\sigma$. The numerator
is bounded from below by a constant independent of $\tilde\sigma$. Averaging over $\tilde\sigma$ concludes the proof of Proposition~\ref{prop:goe}.
\end{proof}

For any two matrices $A$ and $B$ one can find $\phi_0$ and $\psi_0$ so that $\|\phi_0\|=\|\psi_0\|=1$ and
$\|A^*\psi_0\| = \|A\|$, $\|B\phi_0\| = \|B\|$. Then, for $S = (V - \sigma)^{-1}$,
\[ \|ASB\| \geq \left| \langle ASB \phi_0 , \, \psi_0 \rangle \right| = \left| \langle SB \phi_0 , \, A^*\psi_0 \rangle \right|~, \]
and by Proposition~\ref{prop:goe}
\[ \mathbb{E} \|ASB\|^s \geq C^{-1} \|A\|^s \|B \|^s~. \]
Applying this to $A = G^{\mathcal{T}_{u,x}}(0, u_-; z)$, $S = \lambda G^{\mathcal{T}_x}(u,u; z) = (V(x) - \sigma)^{-1}$, and
$B = G^{\mathcal{T}_{u,x}}(u_+, x_-; z)$, we obtain:
\[\begin{split}
&\mathbb{E} \| G^{\mathcal{T}_{u,x}}(0, u_-; z) G^{\mathcal{T}_x}(u,u; z) G^{\mathcal{T}_{u,x}}(u_+, x_-; z)  \|^s \\
&\qquad\geq \mathbb{E} \mathbb{E}_{V(x)} \| G^{\mathcal{T}_{u,x}}(0, u_-; z) G^{\mathcal{T}_x}(u,u; z) G^{\mathcal{T}_{u,x}}(u_+, x_-; z)  \|^s
    \mathbbm{1}_{\|\sigma\| \leq t} \\
&\qquad\geq C_{\lambda,t}^{-1} \mathbb{E}  \| G^{\mathcal{T}_{u,x}}(0, u_-; z) \|^s \| G^{\mathcal{T}_{u,x}}(u_+, x_-; z)  \|^s
    \mathbbm{1}_{\|\sigma\| \leq t} \\
&\qquad\geq C_t^{-1} \mathbb{E}  \| G^{\mathcal{T}_{u,x}}(0, u_-; z) \|^s \| G^{\mathcal{T}_{u,x}}(u_+, x_-; z)  \|^s
    \prod_{w \in \mathcal{N}_u} \mathbbm{1}_{\| G^{\mathcal{T}_{u,x}}(w,w; z) \| \leq C t}~,
\end{split}\]
where we omitted the dependence on $\lambda$ and $W$. This expression is equal to
\begin{multline*} 
C_t^{-1} \left\{ \mathbb{E}  \| G^{\mathcal{T}_{u,x}}(0, u_-; z) \|^s \mathbbm{1}_{\| G^{\mathcal{T}_{u,x}}(u_-,u_-; z) \| \leq C t} \right\}\\
\left\{ \mathbb{E}  \| G^{\mathcal{T}_{u,x}}(u_+, x_-; z)  \|^s \mathbbm{1}_{\| G^{\mathcal{T}_{u,x}}(u_+,u_+; z) \| \leq C t} \right\}
\prod_{w \in \mathcal{N}_u \setminus u_\pm} \left\{ \mathbb{E} \mathbbm{1}_{\| G^{\mathcal{T}_{u,x}}(w,w; z) \| \leq C t} \right\}~.
\end{multline*}
By Chebyshev's inequality, 
\[ \mathbb{E} \mathbbm{1}_{\| G^{\mathcal{T}_{u,x}}(w,w; z) \| \leq C t} \geq 1 - C' t^{-s}\] 
can be made arbitrarily close to $1$ by choosing $t$ large enough. It remains to show that 
\[ \mathbb{E}_{V(w)} \| G^{\mathcal{T}_{u,x}}(w, w'; z)  \|^s \mathbbm{1}_{\| G^{\mathcal{T}_{u,x}}(w,w; z) \| \geq C t}
  \leq \epsilon(t) \mathbb{E}_{V(w)} \mathbb{E}  \| G^{\mathcal{T}_{u,x}}(w, w'; z)  \|^s~, \]
where $\epsilon(t) \to 0$ as $t \to \infty$. We will prove a stronger statement:
\begin{multline*} \mathbb{E}_{V(w)^\mathrm{diag}} \| G^{\mathcal{T}_{u,x}}(w, w'; z)  \|^s \mathbbm{1}_{\| G^{\mathcal{T}_{u,x}}(w,w; z) \| \geq C t}\\
  \leq \epsilon(t) \mathbb{E}_{V(w)^\mathrm{diag}} \mathbb{E}  \| G^{\mathcal{T}_{u,x}}(w, w'; z)  \|^s~, 
\end{multline*}
where $\mathbb{E}_{V(w)^\mathrm{diag}}$ denotes the expectation over the diagonal elements of $V(w)$. Since the dependence on
$W$ is not important for us, it is sufficient to show that, for every $j$ and $k$,
\begin{multline*}
\mathbb{E}_{V(w)^\mathrm{diag}} | G^{\mathcal{T}_{u,x}}(w, w'; z)(j, k) |^s \mathbbm{1}_{\| G^{\mathcal{T}_{u,x}}(w,w; z) \| \geq C t}
  \\
\leq \epsilon(t) \mathbb{E}_{V(w)^\mathrm{diag}} \mathbb{E}  | G^{\mathcal{T}_{u,x}}(w, w'; z) (j,k) |^s~. 
\end{multline*}
Choose $p, q > 1$ so that $1/p + 1/q = 1$ and $sp < 1$. By H\"older's inequality, 
\begin{multline*}
 \mathbb{E}_{V(w)^\mathrm{diag}} | G^{\mathcal{T}_{u,x}}(w, w'; z)(j, k) |^s \mathbbm{1}_{\| G^{\mathcal{T}_{u,x}}(w,w; z) \| \geq C t} 
  \\ 
\leq \left\{ \mathbb{E}_{V(w)^\mathrm{diag}} | G^{\mathcal{T}_{u,x}}(w, w'; z)(j, k) |^{sp}\right\}^{1/p}
  \left\{ \mathbb{E} \mathbbm{1}_{\| G^{\mathcal{T}_{u,x}}(w,w; z) \| \geq C t}  \right\}^{1/q} \\
\leq C' t^{-s/q} \left\{ \mathbb{E}_{V(w)^\mathrm{diag}} | G^{\mathcal{T}_{u,x}}(w, w'; z)(j, k) |^{sp}\right\}^{1/p}~.
\end{multline*}
It remains to show that
\begin{multline}\label{eq:revhoelder}
   \left\{ \mathbb{E}_{V(w)^\mathrm{diag}} | G^{\mathcal{T}_{u,x}}(w, w'; z)(j, k) |^{sp}\right\}^{1/(sp)} \\
  \leq C \left\{ \mathbb{E}_{V(w)^\mathrm{diag}} | G^{\mathcal{T}_{u,x}}(w, w'; z)(j, k) |^{s}\right\}^{1/s}~.  
\end{multline}
The expression $G^{\mathcal{T}_{u,x}}(w, w'; z)(j, k)$ is a fractional-linear function of every diagonal element of
$V(w)$. Therefore (\ref{eq:revhoelder}) follows from the following decoupling lemma
\begin{prop}
Let $X_j$, $1 \leq j \leq W$, be independent identically distributed random variables with bounded density and
finite moments. Then, for every function $f(x_1, \cdots, x_W)$ which is fractional-linear as a function of every variable,
and every $0 < \alpha < \beta < 1$,
\[ (\mathbb{E} |f(X_1, \cdots, X_W)|^\beta)^{1/\beta} \leq C (\mathbb{E} |f(X_1, \cdots, X_W)|^\alpha)^{1/\alpha}~,   \]
where $C>0$ may depend on $\alpha$ and $\beta$ but not on $f$.
\end{prop}
The proof is given (in more general setting) in \cite[Proposition~3.2]{ESS}. This concludes the proof of Lemma~\ref{l:3}.
\end{proof}

Similar considerations allow to extend the arguments leading to two more statements from \cite{AW} to our matrix setting:

\begin{lemma}[Lemma~3.4*]\label{l:3.4} For $s \in (0, 1)$,
\[ \frac{1}{C(s,z)} \leq \frac{\mathbb{E} \|G^{\mathcal{T}_x}(0,x_-; z)\|^s}{\|G^{\mathcal{T}_{x_-}}(0,x_{--}; z)\|^s} \leq C(s,z)~,\]
and
\[ \frac{1}{C(s,z)} \leq \frac{\mathbb{E} \| G(0, x_-; z)\|^s}{\mathbb{E} \|G^{\mathcal{T}_x}(0,x_-;z)\|^s } \leq C(s,z)~,\]
where $C(s, z)$ remainds bounded (for fixed $\Re z$) as $\Im z \to +0$.
\end{lemma}

\begin{prop}[Theorem~3.2*]\label{p:thm:3.2}
Let 
\[ \phi_\lambda(s; z)  = \lim_{\mathrm{dist}(x,0) \to \infty} \log \mathbb{E} \|G_\lambda(0,x; z)\|^s~. \]
For any $z \in \mathbb{C}^+$ the function $(0, \infty) \ni s \mapsto \phi_\lambda(s; z)$ has the following properties:
\begin{enumerate}
 \item $\phi_\lambda(\cdot, z)$ is convex and non-increasing;
 \item for $s \in (0, 2]$, 
\[ -s L(z) \leq \phi_\lambda(s; z) \leq - s \log \sqrt{K}~; \]
 \item for any $s \in (0, 1)$ and $x \in \mathcal{T}$,
\[ \frac{1}{C(s, z)} e^{\phi_\lambda(s; z) \mathrm{dist}(x, 0)} \leq \mathbb{E} \|G_\lambda(0, x; z)\|^s 
  \leq C(s, z) e^{\phi_\lambda(s; z) \mathrm{dist}(x, 0)}~,\]
where $C(s,z) \in (0, \infty)$; if $s \in (0, 1)$, $C(s, z)$ remains bounded as $\Im z \to +0$.
\end{enumerate}
\end{prop}

\begin{df}
The {\em no-a.c.\ hypothesis} holds at energy $E \in \mathbb{R}$ if, for a fixed vector $v$,
\[ \Im \langle (H - E - i0)^{-1} v , v \rangle = 0 \]
almost surely.
\end{df}
Note that the definition does not depend on the choice of the vector $v$.

\begin{cl}\label{cl:6}
Under the no-ac hypothesis $G(0, 0; E+i0)$ is almost surely real symmetric. 
\end{cl}

\begin{proof}
Let us show that
\begin{equation}\label{eq:sa}
G(0,0; E+i0)_{kj} =  \overline{G(0,0; E+i0)_{jk}}~. 
\end{equation}
For $j = k$ this follows drectly from the definition (applied to $v = (0, j)$). For $j \neq k$, apply
the definition to 
\[ v_1 = \delta(0, j) + \delta(0, k)~, \quad v_2 = \delta(0, j) + i \delta(0, k)~. \]
We obtain that
\begin{multline*}\langle G(0,0; E+i0) v_1, v_1 \rangle \\
  = G(0,0; E+i0)_{jj} + G(0,0; E+i0)_{kk} + G(0,0; E+i0)_{jk} + G(0,0; E+i0)_{kj}
\end{multline*}
is real, hence 
\[ G(0,0; E+i0)_{jk} + G(0,0; E+i0)_{kj} \]
is real; also,
\begin{multline*} \langle G(0,0; E+i0) v_2, v_2 \rangle\\
 = G(0,0; E+i0)_{jj} + G(0,0; E+i0)_{kk} - i G(0,0; E+i0)_{jk} + i G(0,0; E+i0)_{kj}
\end{multline*}
is real, hence
\[ G(0,0; E+i0)_{jk} - G(0,0; E+i0)_{kj} \]
is pure imaginary. To conclude the proof of (\ref{eq:sa}), note that if $a + b$ is real and $a - b$ is pure imaginary, then
$a = \bar{b}$. 

$G$ is always symmetric, hence (\ref{eq:sa}) implies that $G(0,0; E+i0)$ is real symmetric.
\end{proof}

\begin{cl}\label{cl:7}
For any real symmetric $W \times W$ matrix $A$,
\[ \| A \| \leq C_W \max \left\{ \max_j |\langle A e_j , e_j \rangle|, \max_{j\neq k} |\langle A (e_j + e_k) , (e_j + e_k) \rangle| \right\}~. \]
\end{cl}

\begin{proof}
Denote $\|A\| = R$. Then $\|A\|_\infty \geq R/B_W$ (where $\|\cdot\|_\infty$ stands for the maximum of the absolute values of
the matrix entries).  There are two cases:
\begin{enumerate}
 \item There exists $j$ so that $|A_{jj}| \geq \frac{R}{3 B_W}$ for some $j$ (then the conclusion of the claim is obvious)
 \item There exist $j$ and $k$ so that $|A_{jk}| \geq \frac{R}{B_W}$, and $|A_{jj}|,|A_{kk}| < \frac{R}{3B_W}$. Then
\begin{multline*}|\langle A (e_j + e_k) , (e_j + e_k) \rangle| \\
    = | a_{jj} + a_{kk} + 2 a_{jk} |\geq 2 |a_{jk}| - |a_{kk}| - |a_{jj}| \geq \frac{R}{B_W}~. 
\end{multline*}
\end{enumerate}

\end{proof}

\begin{prop}\label{prop:4secondmoment}
Under the no-ac assumption, there exists $C > 0$ so that for any $n\geq1$ and $\eta > 0$ 
\[ \mathbb{E} N(N-1) \leq C \tau^{-2} K^{2n}~. \]
\end{prop}

\begin{proof}
Recall that 
\[ E_x = \left\{ | \langle G(x, x; E + i \eta) \tilde{v}, \tilde{v} \rangle| \geq \tau\right\}~,\]
therefore (by Claim~\ref{cl:7})
\[ E_x \subset \tilde{E}_x =
   \left\{ \|  G(x, x; E + i \eta)\| \geq \tau  \right\}
   \subset \bigcup_j \tilde{E}_x^j \, \cup \, \bigcup_{jk} \tilde{E}_x^{jk}~,
\]
where
\[ \tilde{E}_x^j = \left\{ | \langle G(x, x; E + i \eta) e_j, e_j \rangle| \geq \tau/C  \right\} \]
and
\[ \tilde{E}_x^{jk} = \left\{ | \langle G(x, x; E + i \eta) (e_j + e_k), (e_j + e_k) \rangle| \geq \tau/C \right\}~. \]
Therefore
\[\begin{split}
\mathbb{E} N(N-1) 
  &= \sum_{x,y \in S_n, x \neq y} \mathbb{P}(R_x \cap E_x \cap R_y \cap E_y) \\
  &\leq \sum \mathbb{P}(E_x \cap E_y)\\
  &\leq \sum \Big\{ \sum_{jj'} \mathbb{P}(\tilde{E}_x^j \cap \tilde{E}_y^{j'}) 
    + \sum_{jj'k'} \mathbb{P}(\tilde{E}_x^j \cap \tilde{E}_y^{j'k'}) \\
    &\qquad\qquad\qquad+ \sum_{jkj'} \mathbb{P}(\tilde{E}_x^{jk} \cap \tilde{E}_y^{j'}) 
    + \sum_{jkj'k'} \mathbb{P}(\tilde{E}_x^{jk} \cap \tilde{E}_y^{j'k'}) \Big\} \\
  &= \sum (I + II + III + IV).
\end{split}\]
Let us estimate the terms $I$ (the other terms are estimated in the same way). We apply \cite[Theorem~A.2]{AW}. It
yields:
\begin{multline*}
 \mathbb{P} (\tilde{E}_x^j \cap \tilde{E}_y^{j'}) \\
  \leq 
  \frac{C}{\tau} \left\{ \frac{C}{\tau} + \mathbb{E} \min\left(1, 
  \sum_{u \sim (x, j), v \sim (y, j')} \left| H(x,j; u) G^{(x,j; y,j')} (u, v; E+i\eta) H(v; y, j) \right|  \right) \right\}~.
\end{multline*}

Here $ H(x,j; u) $ and $H(v; y, j) $ are Gaussian random variables, independent of each other and of $G^{(x,j; y,j')}$,
the Green function corresponding to the operator obtained by erasing the vertices $(x,j)$ and $(y,j')$ of
$\mathcal{T} \times G$. The first term is of the desired form since the number of addends is bounded by $C_W K^{2n}$. 
For the second term we use the inequality 
\[ \min(1, |x|) \leq |x|^s~, \quad 0 \leq s \leq 1~, \]
and then estimate:
\[\begin{split}
&\mathbb{E} \sum_{u \sim (x, j), v \sim (y, j')} \left| H(x,j; u) G^{(x,j; y,j')} (u, v; E+i\eta) H(v; y, j) \right|^s\\
&\quad=  \sum_{u \sim (x, j), v \sim (y, j')}  \mathbb{E} |H(x,j; u)|^s \, \mathbb{E} |G^{(x,j; y,j')}(u, v; E+i\eta)|^s \,
  \mathbb{E} |H(v; y, j)|^s \\
&\quad\leq C \sum_{u \sim (x, j), v \sim (y, j')} \mathbb{E} |G^{(x,j; y,j')}(u, v; E+i\eta)|^s~.
\end{split}\]
If $u = (x, k)$, $v = (y,k')$ (where $k \neq j$, $k' \neq j'$), repeated application of Lemma~\ref{l:3.4} and
Proposition~\ref{p:thm:3.2} yields
\[\begin{split}
\mathbb{E} |G^{(x,j; y,j')}(u,v; E+i\eta)|^s
    &\leq C \mathbb{E} \|G(x,y;E+i\eta)\|^s \leq C' K^{- \frac{s}{2} \mathrm{dist}(x, y)}~.
\end{split}\]

Combining these estimates and taking $s = \frac{L(E)+2\delta}{\log K} \in (0, 1)$. we obtain
the desired bound. This completes the proof of Proposition~\ref{prop:4secondmoment}.

\end{proof}

\begin{prop}[Modified Theorem~4.6*]\label{prop:thm:4.6'*}
For almost all 
\[ E \in \sigma(H) \cap \left\{ L(E) < \log K \right\} \cap \left\{ \text{no-ac holds} \right\}~,\]
there exist $\delta, p_0 > 0$ and $n_0 \geq 0$ so that for all $n \geq n_0$
\[ \liminf_{\eta \to 0} \mathbb{P} \left\{ \max_{x \in S_n} \|G(0, x; E+i\eta\| \geq e^{\delta n} \right\} \geq p_0~. \]
\end{prop}

\begin{proof}
By Proposition~\ref{prop:2firstmoment} and Proposition~\ref{prop:4secondmoment} there exist $C$,$\eta_0$ and  $n_0$ so that for
$n \geq n_0$ and $\eta \in (0, \eta_0)$
\[ \frac{\mathbb{E} N^2}{\left\{ \mathbb{E} N \right\}^2}
  = \frac{1}{\mathbb{E} N} + \frac{\mathbb{E} N(N-1)}{\left\{ \mathbb{E} N \right\}^2} \leq C~.\]
Therefore
\[ \mathbb{P} \left\{ N \geq 1 \right\} \geq \frac{\left\{ \mathbb{E} N\right\}^2}{\mathbb{E} N^2} \geq \frac{1}{C} \]
uniformly in $n \geq n_0$ and $\eta \in (0, \eta_0)$. 
\end{proof}

\begin{proof}[Proof of Theorem~\ref{thm2'}] We argue by contradiction: if the no-ac hypothesis holds for a given
$E \in \sigma(H)$, the conclusion of Proposition~\ref{prop:thm:4.6'*} implies that
\[ \sum \|G(0,x; E+i0)\|^2 = \infty \]
with positives probability and hence almost surely. Proposition~\ref{prop1:sw} conludes the proof.
\end{proof}

\section{Proof of Theorem~\ref{thm2}}

Denote
\[ \Gamma(y) = \Gamma(y; E+i\eta) = G^{\mathcal{T}_{y_-}}(y,y; E+i\eta)~;
\quad \tilde\Gamma(y) = \frac{\Gamma(y) - \Gamma(y)^*}{2i} \]
(the latter is the matrix analogue of $\Im \Gamma$ from \cite{AW}). Theorem~\ref{thm2} will follow
from the following statements:

\begin{lemma}[Lemma~4.4*]\label{l:4.4*}
For any $A > 0$, if 
\[ \mathbb{P}  \left\{ \|\tilde\Gamma\| \geq A \right\} \geq q > 0 \]
for some $q \in (0, 1)$, then 
\[ \mathbb{P} \left\{ \|\tilde\Gamma\| \geq \frac{A}{R} \right\} \to 1 \]
as $R \to \infty$, uniformly in $\eta > 0$.
\end{lemma}

The proof is identical to that of \cite[Lemma~4.4]{AW} (note however that, unlike the rest
of the current paper, one has to work with the fastest Lyapunov exponent rather than the slowest one).

\begin{prop}[Theorem~4.6*]\label{prop:thm:4.6*}
For almost all 
\[ E \in \sigma(H) \cap \left\{ L(E) < \log K \right\} \cap \left\{ \text{no-ac holds} \right\}~,\]
there exist $\delta, p_0 > 0$ and $n_0 \geq 0$ so that for all $n \geq n_0$ 
\begin{multline*}
 \liminf_{\eta \to +0} \mathbb{P} \big\{ 
  \exists x \in S_n, y \in \mathcal{N}_x^+ \, \big| \, 
  \|G^{\mathcal{T}_x}(0, x_-, E+i\eta)\| \geq e^{-(L(E)+\delta)n}~, \,
  \|\tilde\Gamma\| \geq \xi(p)~, \,\\
  \left| \left\langle G(x,x; E+i\eta) w_{\max}(\tilde\Gamma(y)), 
       w_{\max}(G^{\mathcal{T}_x}(0,x_-; E+i\eta)^* G^{\mathcal{T}_x}(0,x_-; E+i\eta)) 
      \right\rangle\right| \\ \geq e^{+(L(E)+2\delta)n} \Big\}  
  \geq q > 0~,
\end{multline*}
where
\begin{enumerate}
 \item $q$ may depend on $\delta$ and $p$, but not on $\eta$ and $n$;
 \item $\xi(p) = \inf \left\{ t  \, \mid \, \mathbb{P}\{\|\tilde\Gamma\| \geq t \} \geq p \right\}$ is  the $p$-th quantile of $\|\tilde\Gamma\|$;
 \item $w_{\max}$ denotes the eigenvector asociated with the maximal eigevalue.
\end{enumerate}
\end{prop}

The following lemma will be used both in the proof and in the application of Proposition~\ref{prop:thm:4.6*}.
\begin{lemma}\label{l:**}
The (self-adjoint) matrix $\tilde\Gamma(0)$ admits the lower bound
\[ \tilde\Gamma(0) \geq \sum_{x \in S_n} \sum_{y \in \mathcal{N}_x^+}  G(0, x; E+i\eta) \tilde\Gamma(y) G(0, x;E+i\eta)^* \]
in the sense of quadratic forms.
\end{lemma}

\begin{proof}
From the resolvent identity,
\[\begin{split} 
\tilde\Gamma(0) &= \frac{\Gamma(0) - \Gamma(0)^*}{2i}    \\
  &= \frac{1}{2i} \Gamma(0) \left\{ \eta + \sum_{y \in \mathcal{N}_x^+} (\Gamma(y) - \Gamma(y)^*) \right\} \Gamma(0)^*\\
  &\geq \sum_{y \in \mathcal{N}_x^+} \Gamma(0)  (\Gamma(y) - \Gamma(y)^*) \Gamma(0)^* \\
  &= \sum_{y \in \mathcal{N}_x^+} G(0, 0; E+i\eta)  (\Gamma(y) - \Gamma(y)^*) G(0, 0; E+i\eta)^*~. 
\end{split}\] 
This yields the statement for $n = 0$. The statement for larger $n$ follows by iteration.
\end{proof}

\begin{proof}[Proof of Proposition~\ref{prop:thm:4.6*}]
Denote 
\[\begin{split}
I_x &= \left\{ \|\tilde\Gamma(x)\| \geq \xi(p) \right\}~, \\
R_x &= \left\{ \|G^{\mathcal{T}_x}(0,x_-;E+i\eta)\| \geq e^{-(L(E)+\delta)n}\right\}~, \\
E_x &= \left\{ \left| \left\langle G(x,x; E+i\eta)v; w\right\rangle \right| \geq \tau \right\}~,
\end{split}\]
where 
\[ v = w_{\max}(\tilde\Gamma(y))~, \, w = w_{\max}(G^{\mathcal{T}_x}(0,x_-; E+i\eta)^* G^{\mathcal{T}_x}(0,x_-; E+i\eta))~.\]
Then Proposition~\ref{prop:thm:4.6*} states that
\[ \liminf_{\eta \to +0} \mathbb{P} \left\{ \bigcup_x I_x \cap R_x \cap E_x \right\} \geq q > 0~. \]
Denote 
\[ N = \sum_{x \in S_n} \mathbbm{1}_{I_x \cap R_x \cap E_x}~. \]
As in the proof of Theorem~\ref{thm2'}, we shall prove that
\[ \frac{\mathbb{E} N^2}{\left\{ \mathbb{E} N \right\}^2} \leq C~. \]
The upper bound on $\mathbb{E} N(N-1)$ follows from the argument of Proposition~\ref{prop:4secondmoment}. Indeed,
in the notation of the proof of Proposition~\ref{prop:4secondmoment},
\[ \mathbb{P} (I_x \cap R_x \cap E_x \cap I_y \cap R_y \cap E_y)
  \leq \mathbb{P}(E_x \cap E_y) \leq \mathbb{P} (\tilde{E}_x \cap \tilde{E}_y)~,
\]
hence 
\[ \mathbb{E} N(N-1) \leq C \tau^{-2} K^{2n}~. \]
To bound $\mathbb{E} N$ from below, we need to show that
\[ \mathbb{P}(I_x \cap R_x \cap E_x) \geq C \tau^{-1}~. \]
By the parallelogram law,
\[ \langle G(x,x) v, w \rangle = 
  \frac{1}{4} \left[ \langle G(x,x)(v+w), (v+w) \rangle - \langle G(x,x)(v-w), (v-w) \rangle\right]~. \]
In our case, $\| v \| = \|w \| = 1$, hence $v+w \perp v-w$. Without loss of generality we may assume
that $\| v + w \| \geq \|v-w\|$, then $\|v+w\| \geq 2 \geq \|v-w\|$. Set $e_1 = (v+w)/\|v+w\|$, 
$e_2 = (v-w)/\|v-w\|$. Then
\[ \left\{ |\langle G(x,x) v, w \rangle| \geq \tau \right\} \supset
   \left\{ |\langle G(x, x) e_1, e_1 \rangle | \geq 2\tau, \, |\langle G(x, x) e_2, e_2 \rangle | \leq \tau \right\}~.
\]
No generality is lost if we assume that $e_1$ and $e_2$ are the first two vectors of the standard basis.
Let $P$ be the projection onto $e_1, e_2$. Then
\[ \left(\begin{array}{cc} G_{11} & G_{12} \\ G_{12} & G_{22} \end{array}\right)
  = P G(x, x) P = \left( \lambda \left(\begin{array}{cc} V_{11} & 0 \\ 0 & V_{22}\end{array}\right) - X \right)^{-1}~,\]
where
\[ X = \left(\begin{array}{cc} a & b \\ b & c \end{array}\right) \]
is independent of $V_{11}$ and $V_{22}$. Consider two cases:
\begin{enumerate}
 \item $|b| \leq 1/\sqrt{\tau}$. Then the argument of Proposition~\ref{prop:2firstmoment} yields
\[ \mathbb{P}_{V_{11},V_{22}} \left\{ |G_{11}|\geq 2\tau \right\} \geq \frac{1}{C\tau}~, \]
whereas \cite[Theorem~A.2]{AW} yields
\[ \mathbb{P}_{V_{11},V_{22}} \left\{ |G_{11}|\geq 2\tau~, \, |G_{22}|\geq \tau \right\} \leq \frac{C}{\tau^{3/2}}~. \]
Therefore
\[ \mathbb{P}_{V_{11},V_{22}} \left\{ |G_{11}|\geq 2\tau~, \, |G_{22}|\leq \tau \right\} \geq \frac{1}{C'\tau}~. \]
\item $|b| > 1/\sqrt{\tau}$. If $|G_{22}|\geq \tau$, then 
\[ |(V_{11}-a) - b^2/(V_{22} - c)| = |G_{22}|^{-1} \leq \frac{1}{\tau}~, \]
therefore 
\[ |V_{11} - a| \leq \frac{1}{\tau} + \left| \frac{b^2}{V_{22}-c} \right|~.\] 
If in addition $|V_{22} - c| > 2b$, then
\[ |V_{11} - a| \leq \frac{b}{2} + \frac{1}{\tau} \leq \frac{2b}{3}~. \]
Therefore
\[ \left| \frac{V_{22}-c}{V_{11}-a} \right| \geq \frac{2b}{2b/3} = 3~. \]
This implies
\[\begin{split}
\frac{1}{|G_{22}|}
  &= \left| V_{22} - c - \frac{b^2}{V_{11}-a} \right| \\
  &= \left| \frac{V_{22}-c}{V_{11}-a} \right| \left| V_{11} - a - \frac{b^2}{V_{22}-c} \right| 
  \geq \frac{3}{|G_{11}|}~.
\end{split}\]
Hence in this case
\begin{multline*}
\left\{ V_{11}, V_{22} \, \mid \, |G_{11}| \geq 2\tau, |G_{22}| \leq \tau \right\} \\
 \supset \left\{ V_{11}, V_{22}\, \mid \, \frac{1}{3\tau} < \frac{1}{|G_{11}|} < \frac{1}{2\tau}~, \, |V_{22}| > 2b \right\}~,
\end{multline*}
and the probability of this event is again $\geq C^{-1}(b) \tau^{-1}$. The rest of the argument 
follows the proof of Proposition~\ref{prop:2firstmoment}.

\end{enumerate}

\end{proof}

\begin{proof}[Proof of Theorem~\ref{thm2}]
Theorem~\ref{thm2} follows immediately from Lemma~\ref{l:4.4*}, Proposition~\ref{prop:thm:4.6*}, and
Lemma~\ref{l:**}.
\end{proof}

\subsection*{Acknowledgement} I am grateful to Michael Aizenman and to Simone Warzel for numerous helpful
conversations, in particular, for the explanations pertaining to the work~\cite{AW}.

\end{document}